\documentclass[runningheads]{llncs}

\author{
    Mathias Hall-Andersen
    \and
    Nikolaj I. Schwartzbach
}

\institute{
    Department of Computer Science, Aarhus University
}

\title{Game theory on the blockchain: a model for games with smart contracts}

\usepackage{adjustbox}
\usepackage{amsmath,amssymb}
\usepackage{mathrsfs}
\usepackage{framed}
\usepackage{booktabs}
\usepackage{algorithmic}
\usepackage{bbm}
\usepackage{xcolor}
\usepackage{hyperref}
\usepackage[inline,shortlabels]{enumitem}
\usepackage[ruled,vlined]{algorithm2e}
\usepackage[capitalise,nameinlink]{cleveref}
\usepackage{todonotes}

\usepackage{tikz}
\usepackage{tikz-qtree}

\usetikzlibrary{calc}
\usetikzlibrary{trees}
\usetikzlibrary{shapes.geometric}
\usetikzlibrary{arrows.meta}

\newcommand{\inducible}[0]{\mathscr{R}}

\newcommand{\coloring}{\textsc{3-coloring}}

\usepackage[ruled,vlined]{algorithm2e}

\DeclareMathOperator*{\argmax}{argmax}

\renewcommand{\P}{\mathsf{P}}

\newcommand{\Pcolor}{P_\text{color}}
\newcommand{\Pcheck}{P_\text{check}}

\newcommand{\player}[0]{P}
\newcommand{\playern}[1]{\player_{#1}}

\newcommand{\gamepath}[0]{\mathfrak{p}} 
\newcommand{\subgame}[0]{G_s} 

\newcommand{\utility}[0]{u}

\newcommand{\gametree}[0]{T}
\newcommand{\cut}[0]{c}
\newcommand{\cutof}[1]{\ensuremath{c^{(#1)}}}

\newcommand{\playerL}[0]{\playerA}
\newcommand{\playerF}[0]{\playerB}

\newcommand{\playerA}[0]{\playern{1}}
\newcommand{\playerB}[0]{\playern{2}}

\newcommand{\EXP}{\textsf{EXP}}

\def\btc{}

\newcommand{\PSPACE}{\textsf{PSPACE}}
\renewcommand{\P}{\textsf{P}}
\newcommand{\NAND}{\textsf{NAND}}
\newcommand{\NP}{\textsf{NP}}

\hypersetup{nolinks=true}

\begin{document}

\renewcommand{\floatpagefraction}{.8}%
\SetAlgoLined
\SetNoFillComment
\maketitle

\begin{abstract}
    We propose a model for games in which the players have shared access to a blockchain that allows them to deploy smart contracts to act on their behalf.
    This changes fundamental game-theoretic assumptions about rationality since a contract can commit a player to act irrationally in specific subgames,
    making credible otherwise non-credible threats. This is further complicated by considering the interaction between multiple contracts which can reason about each other.
    This changes the nature of the game in a nontrivial way as choosing which contract to play can itself be considered a move in the game.
    Our model generalizes known notions of equilibria,
    with a single contract being equivalent to a Stackelberg equilibrium, and two contracts being equivalent to a reverse Stackelberg equilibrium.
    We prove a number of bounds on the complexity of computing SPE in such games with smart contracts.
    We show that computing an SPE is $\PSPACE$-hard in the general case.
    Specifically, in games with $k$ contracts, we show that computing an SPE is $\Sigma_k^\P$-hard for games of imperfect information.
    We show that computing an SPE remains $\PSPACE$-hard in games of perfect information if we allow for an unbounded number of contracts.
    We give an algorithm for computing an SPE in two-contract games of perfect information that runs in time $O(m\ell)$ where $m$ is the size of the game tree and $\ell$ is the number of terminal nodes.
    Finally, we conjecture the problem to be $\NP$-complete for three contracts.
\end{abstract}

\section{Introduction}

This paper is motivated by the games that arise on permissionless blockchains such as Ethereum \cite{yellowpaper} that offer ``smart contract'' functionality:
in these permissionless systems, parties can deploy smart contracts without prior authorization by buying the ``tokens'' required to execute the contract. By smart contracts, we mean arbitrary pieces of code written in a Turing-complete language\footnote{However the running time of the contracts is limited by the execution environment.} capable of maintaining state (including funds) and interact with other smart contracts by invoking methods on them. Essentially, smart contracts are objects in the Java sense.
Parties can also invoke methods on the smart contracts manually.
Note that the state of all smart contracts is public and can be inspected by any party at any time. This changes fundamental game-theoretic assumptions about rationality: in particular, it might be rational for a player to deploy a contract that commits them to act irrationally in certain situations to make credible otherwise non-credible threats. This gives rise to very complex games in which parties can commit to strategies,
that in turn depend upon other players' committed strategies. Reasoning about such equilibria is important when considering games that are meant to be played on a blockchain, since the players - at least in principle - always have the option of deploying such contracts. In the literature, this is known as a Stackelberg equilibrium where a designated leader commits to a strategy before playing the game. In general, because of first-mover advantage, being able to deploy a contract first is never a disadvantage, since a player can choose to deploy the empty contract that commits them to nothing. It is well-known that it is hard to compute the Stackelberg equilibrium in the general case \cite{letchford_phd}, though much less is known about the complexity when there are several of these contracts in play: when there are two contracts, the first contract can depend on the second contract in what is known as a reverse Stackelberg equilibrium \cite{reverse_stackelberg_closed_loop_2,reverse_stackelberg_closed_loop_1,reverse_stackelberg_threat}. This is again strictly advantageous for the leader since they can punish the follower for choosing the wrong strategy. In this paper, we present a model that generalizes (reverse) Stackelberg games, that we believe captures these types of games
and which may be of wider interest.
In practical terms, we believe that our model is of interest when analyzing distributed systems for "game-theoretic security"
in settings where the players naturally have the ability to deploy smart contracts. Potential examples include
proof-of-stake blockchains themselves and financial applications that build upon these systems.

\begin{figure}\begin{adjustbox}{center}
		\begin{tabular}{ c  c  c c | c  c }
			\toprule
			Contracts & Players & ~Information~  & ~Strategies~ & Lower bound & Upper bound \\
			\toprule
			0   & 2     & perfect & pure & \P-hard~\cite{spe_p_complete} & $O(m)$ \cite{osborne} \\\midrule
			0   & 2     & imperfect & mixed & \multicolumn{2}{c}{\textsf{PPAD}-complete~\cite{papadimitriou,chen06}} \\\toprule
			1   & 2     & perfect & pure & \P-hard~\cite{spe_p_complete} & $O(m\ell)$~\cite{kristoffer_spe} \\\midrule
			1   & 2     & perfect & mixed & \multicolumn{2}{c}{\NP-complete~\cite{stackelberg_np_hard}} \\\midrule
			1   & 2     & imperfect & - & \multicolumn{2}{c}{\NP-complete~\cite{stackelberg_np_hard}} \\\toprule
			2   & 2     & perfect & pure & \P-hard~\cite{spe_p_complete} & $O(m\ell)$ [\cref{thm:two_contracts_feasibility}] \\\toprule
			3   & 3     & perfect & pure & Conjectured \NP-hard & \NP~  [\cref{thm:two_contracts_feasibility}] \\\toprule
			$k$ & $2+k$ & imperfect & pure & $\Sigma^{p}_k$-hard [\cref{thm:poly_hierarchy-hard}] & ? \\\toprule
			unbounded & - & perfect & pure & ~~$\PSPACE$-hard [\cref{thm:perfect_information_pspace_hard}] & ? \\
			\bottomrule
		\end{tabular}
	\end{adjustbox}
	\caption{An overview of some existing bounds on the complexity of computing an SPE in extensive-form games and where our results fit in. Here, $m$ is the size of the tree, and $\ell$ is the number of terminal nodes.}
\end{figure}

\subsubsection{Our results}
We propose a game-theoretic model for games in which players have shared access to a blockchain that allows the players to deploy smart contracts to act on their behalf in the games. Allowing a player to deploy a smart contract corresponds to that player making a `cut' in the tree, inducing a new expanded game of exponential size containing as subgames all possible cuts in the game. We show that many settings from the literature on Stackelberg games can be recovered as special cases of our model, with one contract being equivalent to a Stackelberg equilibrium, and two contracts being equivalent to a reverse Stackelberg equilibrium.
We prove bounds on the complexity of computing an SPE in these expanded trees.
We prove a lower bound, showing that computing an SPE in games of imperfect information with $k$ contracts is $\Sigma_k^\P$-hard by reduction from the true quantified Boolean formula problem.
For $k=1$, it is easy to see that a contract can be verified in linear time, establishing \NP-completeness. In general, we conjecture $\Sigma_k^\P$-completeness for games with $k$ contracts, though this turns out to reduce to whether or not contracts can be described in polynomial space.
For games of perfect information with an unbounded number of contracts, we also establish \PSPACE-hardness from a generalization of \coloring. We show an upper bound for $k=2$ and perfect information, namely that computing an SPE in a two-contract game of size $m$ with $\ell$ terminal nodes (and any number of players) can be computed in time $O(m\ell)$. For $k=3$, the problem is clearly in $\NP$ since we can verify a witness using the algorithm for $k=2$, and we conjecture the problem to be \NP-complete. Finally, we discuss various extensions to the model proposed and leave a number of open questions.

\section{Games with smart contracts}


In this section, we give our model of games with smart contracts.
We mostly assume familiarity with game theory and refer to \cite{osborne} for more details.
For simplicity of exposition, we only consider a somewhat restricted class of games, namely finite games in extensive form, and consider only pure strategies in these games.
In addition, we will assume games are in \emph{generic form}, meaning the utilities of all players are unique.
This has the effect that the resulting subgame perfect equilibrium is unique.
Equivalently, we use a tie breaking algorithm to decide among the different subgame perfect equilibria, and slightly perturb the utilities of the players to match the subgame perfect equilibrium chosen by the tie breaker.

Formally, an \emph{extensive-form game $G$} is a finite tree $T$.  We denote by $L \subseteq T$ the set of leaves in $T$, i.e. nodes with no children, and let $m$ denote the number of nodes in $T$. Each leaf $\ell$ is labeled by a vector $u(\ell) \in \mathbb{R}^n$ that denotes the utility $u_i(\ell)$ obtained by party $P_i$ when terminating in the leaf $\ell$. In addition, the game consists of a finite set of $n$ players. We consider a fixed partition of the non-leaves into $n$ sets, one for each player. The game is played by starting at the root, letting the player who owns that node choose a child to recurse into, this is called a move.
We proceed in this fashion until we reach a leaf and distribute its utility vector to the players.
When there is perfect information, a player always knows exactly which subgame they are playing, though more generally we may consider a partition of the non-leafs into \emph{information sets}, where each player is only told the information set to which their node belongs. When all information sets are singletons we say the game has perfect information.
The players are assumed to be \emph{rational}, that is they choose moves to maximize their utility: we say a strategy for each player (a strategy profile) constitutes a \emph{(Nash) equilibrium} if no unilateral deviation by any party results in higher utility for that party. Knowing the other players are rational, for games of perfect information, at each branch a player can anticipate their utility from each of its moves by recursively determining the moves of the other parties.
This process is called \emph{backward induction}, and the resulting strategy profile is a \emph{subgame perfect equilibrium}. A strategy profile is an SPE if it is an equilibrium for every subgame of the game. For games of perfect information, computing the SPE takes linear time in the size of the tree and can be shown to be $\P$-complete \cite{spe_p_complete}.
Later, we will show a lower bound, namely that adding a contract to the tree moves this computation up (at least) a level in the polynomial hierarchy.
Specifically, we show that computing the SPE in $k$-contract games is $\Sigma_k^\P$-hard in the general case with imperfect information.

\subsection{Smart contract moves}

We now give our definition of smart contracts in the context of finite games.
We add a new type of node to our model of games, a \emph{smart contract move}.
Intuitively, whenever a player has a smart contract move, they can deploy a contract that acts on their behalf for the rest of the game.
The set of all such contracts is countably infinite, but fortunately, we can simplify the problem by considering equivalence classes of contracts which ``do the same thing''.
Essentially, the only information relevant to other players is whether or not a given action is still possible to play:
it is only if the contract dictates that a certain action cannot be played, that we can assume a rational player will not play it.
In particular, any contract which does not restrict the moves of a player is equivalent to the player not having a contract.
Such a restriction is called a \emph{cut}.
A cut $\cutof{i}$ for player $\playern{i}$ is defined to be a
union of subtrees whose roots are childen of $\playern{i}$-nodes, such that:
\begin{enumerate*}[(1)]
    \item every node in $\gametree \setminus \cutof{i}$ has a path going to a leaf; a cut is not allowed to destroy the game by removing all moves for a player, and
    \item $\cutof{i}$ respects information sets, that is it `cuts the same' from each node in the same information set.
\end{enumerate*}

\begin{figure}
    \centering
    \begin{tikzpicture}
    [
    level 1/.style={sibling distance=18mm},
    level 2/.style={sibling distance=14mm},
    level 3/.style={sibling distance=14mm},
    level distance=1cm,align=center]
    \node[draw] {$\playerL$}
    child {
        node[draw,circle] {$\playerF$}
        child {
            node[draw, circle] {$\playerL$}
            child { node {$(-\infty,-\infty)$} }
            child { node {$(0,0)$} } }
        child {node {$(1\btc,\,-1\btc)$} }
    };
    \node (A) at (1.85,-1.5) {=};
    \end{tikzpicture}
    \adjustbox{width=0.6\linewidth}{
        \begin{tikzpicture}
        [level 1/.style={sibling distance=25mm},
        level 2/.style={sibling distance=14mm},
        level distance=1cm,align=center,
        every node/.style={thin},
        emph/.style={edge from parent/.style={ very thick,draw}},
        norm/.style={edge from parent/.style={solid,black,thin,draw}}]
        \node[draw,circle] {$\playerL$}
        child[norm] {
            node[draw,circle] {$\playerF$}
            child[emph] {
                node[draw, circle] {$\playerL$}
                child[norm] { node {$(-\infty,-\infty)$} }
                child { node {$(0,0)$} } }
            child[norm] {node {$(1\btc,\,-1\btc)\quad$} }
        }
        child {
            node[draw,circle] {$\playerF$}
            child[emph] {
                node[draw, circle] {$\playerL$}
                child { node {$(0,0)$} } }
            child[norm] {node {$(1\btc,\,-1\btc)\quad$} }
        }
        child[emph] {
            node[draw,circle] {$\playerF$}
            child[norm] {
                node[draw, circle] {$\playerL$}
                child { node {$(-\infty,-\infty)$} }}
            child[emph] {node {$(1\btc,\,-1\btc)$} }
        };
        \end{tikzpicture}
    }
    \caption{
        Expanding a smart contract node for a simple game.
        The square symbol is a smart contract move for player $\playerL$.
        We compute all $\playerL$-cuts in the game and connect them with a node belonging to $\playerL$.
        The first coordinate is the leader payoff, and the second is the follower payoff.
        The dominating paths are shown in bold.
        We see that the optimal strategy for $\playerL$ is to commit to choosing $(-\infty, -\infty)$ 
        unless $\playerF$ chooses $(1\btc, -1\btc)$. 
    }
    \label{fig:expansion}
\end{figure}
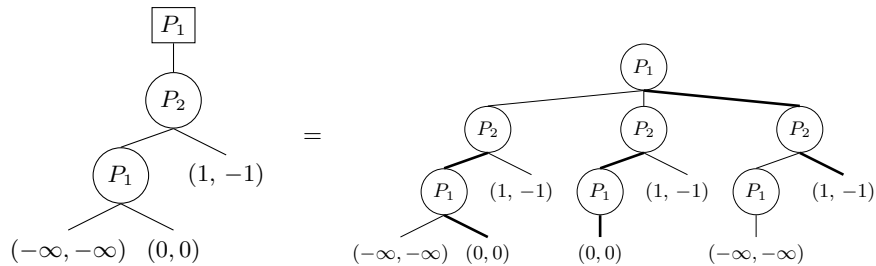

In other words, deploying a smart contract corresponds to choosing a cut in the game tree.
This means that a smart contract node for player $\playern{i}$ in a game $\gametree$ is essentially syntactic sugar for the \emph{expanded tree}
that results by applying the set of all cuts $\cutof{i}$ to $\gametree$ and connecting the resulting games with a new node belonging to $\playern{i}$ at the top.
Computing the corresponding equilibrium with smart contracts then corresponds to the SPE in this expanded tree.
Note that this tree is uniquely determined. See \cref{fig:expansion} for an example.
We use the square symbol in figures to denote smart contract moves. When a game contains multiple smart contract moves, we expand the smart contract nodes recursively in a depth-first manner using the transformation described above.

\subsection{Contracts as Stackelberg equilibria}
As mentioned earlier, the idea to let a party commit to a strategy before playing the game is not a new one: in 1934, von Stackelberg proposed a model for the interaction of two business firms with a designated market leader \cite{stackelberg}.
The market leader holds a dominant position and is therefore allowed to commit to a strategy first, which is revealed to the follower who subsequently decides a strategy. The resulting equilibrium is called a Stackelberg equilibrium.
In this section we show that the Stackelberg equilibrium for a game with leader $\playerL$ and follower $\playerF$ can be recovered as a special case of our model where $\playerL$ has a smart contract.
We use the definition of strong Stackelberg equilibria from \cite{Leitmann1978,Breton1988}.
We note that since the games are assumed to be in generic form,
the follower always has a unique response,
thus making the requirement that the follower break ties in favor of the leader unnecessary.

Let $T$ be a game tree. A \emph{path} $\mathbf{p} \subseteq T$ is a sequence of nodes such that for each $j$, $\mathbf{p}_{j+1}$ is a child of $\mathbf{p}_j$. If $\mathbf{p}$ is a path, we denote by $\mathbf{p}^{(i)} \subseteq \mathbf{p}$ the subset of nodes owned by player $P_i$. Now suppose $T$ has a horizon of $h$. We let $\gamepath=(\gamepath_j)_{j=1}^h \subseteq T$ denote the \emph{dominating path} of the game defined as the path going from the root $\gamepath_1$ to the terminating leaf $\gamepath_h$ in the SPE of the game. 

\begin{definition}
    Let $i \in [n]$ be the index of a player, and let $f(s_i)$ be the best response to $s_i$ for players other than $P_i$. We say $(s_i^*, f(s_{i}^*))$ is a \emph{Stackelberg equilibrium with leader $P_i$} if the following properties hold true:
    \begin{itemize}
        \item \emph{Leader optimality}. For every leader strategy $s_i$, $u_i(s_i^*, f(s_i^*)) \geq u_i(s_i, f(s_i))$.
        \item \emph{Follower best response.} For every $j\neq i$, and every $s_{-i}$, $u_j(s_i^*, f(s_i^*)) \geq u_j(s_i^*, s_{-i})$.$\hfill\diamond$
    \end{itemize}
\end{definition}

\begin{proposition}
The Stackelberg equilibrium with leader $P_i$ is equivalent to $P_i$ having a smart contract move.
\end{proposition}
\begin{proof}
    We show each implication separately:
    \begin{description}
        \item[$\Rightarrow$]
            SPE in the expanded tree $\gametree$ induces a Stackelberg equilibrium in the corresponding Stackelberg game where $P_i$
            commits to all moves in $\gamepath^{(i)}$.
            It is not hard to see that the follower best response $f(s_i^*)$ is defined by the SPE of the
            subgame arising after $\playern{i}$ makes the move $\gamepath_1$ choosing the contract in $\gametree$.
        \item[$\Leftarrow$]
            A Stackelberg equilibrium induces a SPE in the expanded tree $\gametree$ with the same utility:
            let $(s_i^*, f(s_i^*))$ be a Stackelberg equilibrium,
            observe that $s_i^*$ corresponds to a cut $\cut^{(i)} \subseteq \gametree$
            where $\playern{i}$ cuts away all nodes in $\gametree$ not dictated by $s_i^*$. 
            By letting the first move $\gamepath_1$ of $\playern{i}$ correspond to $\cut^{(i)}$,
            the best follower response $f(s_i^*)$ is the SPE in the resulting subgame, and
            hence $\utility(\gamepath) = \utility(s_i^*, f(s_i^*))$.\qed
    \end{description}
\end{proof}

\subsubsection{Multi-leader/multi-follower contracts} Several variants of the basic Stackelberg game has been considered in the literature with multiple leaders and/or followers \cite{multiple_leaders,liu1998stackelberg}. We can model this using smart contracts by forcing some of the contracts to independent of each other: formally, we say a contract is \emph{independent} if it makes the same cut in all subgames corresponding to different contracts. It is not hard to see that multiple leaders can be modelled by adding contracts for each leader, where the contracts are forced to be independent. $\hfill\diamond$

\subsubsection{Reverse Stackelberg contracts}

The reverse Stackelberg equilibrium is an attempt to generalize the regular Stackelberg equilibrium: here, the leader does not commit to a specific strategy \emph{a priori}, rather they provide the follower with a mapping $f$ from follower actions to best response leader actions, see e.g. \cite{reverse_stackelberg_def2,reverse_stackelberg_def} for a definition in the continuous setting. When the follower plays a strategy $s_{-i}$, the leader plays $f(s_{-i})$. This is strictly advantageous for the leader since as pointed out in \cite{reverse_stackelberg_threat}, they can punish the follower for choosing the wrong strategy.

In the following, if $\mathbf{p}$ is a path of length $\ell$, we denote by $\subgame(\mathbf{p})$ the subgame whose root is $\mathbf{p}_\ell$.

\begin{definition}
    Let $i$ be the index of the leader, and $-i$ the index of the follower. We say $(f(s_{-i}^*), s_{-i}^*)$ is a \emph{reverse Stackelberg equilibrium with leader $i$}
    if the following holds for every leader strategy $s_i$ and follower strategy $s_{-i}$, it holds:
    \begin{itemize}
        \item \emph{Leader best response}: $u_i(f(s_{-i}^*), s_{-i}^*) \geq u_i(s_i, s_{-i}^*)$.
        \item \emph{Follower optimality}:  $u_{-i}(f(s_{-i}^*), s_{-i}^*) \geq u_{-i}(f(s_{-i}),s_{-i})$.
            $\hfill\diamond$
    \end{itemize}
\end{definition}

\begin{proposition}
The reverse Stackelberg equilibrium for a two-player game with leader $P_i$ is equivalent to adding two smart contract moves to the game, one for $P_i$, and another for $P_{-i}$ (in that order).
\end{proposition}
\begin{proof}
    We show each implication separately:
    \begin{description}
        \item[$\Rightarrow$]
            The SPE in the expanded tree induces a reverse Stackelberg equilibrium:
            for every possible follower strategy $s_{-i}$, we define $f(s_{-i})$ as the leader strategy in the SPE in the subgame $\subgame(\langle \gamepath_1, s_{-i} \rangle)$ after the two moves, where we slightly abuse notation to let $s_{-i}$ mean that $P_{-i}$ chooses a cut where their SPE is $s_{-i}$.
            Leader best response follows from the observation that $\gamepath_1$
            corresponds to the optimal set of cuts of $\playern{i}$ moves in response to every possible cut of of $\playern{-i}$ moves.

        \item[$\Leftarrow$]
            A reverse Stackelberg equilibrium induces an SPE in the expanded tree:
            let $(f(s^*_{-i}), s_{-i}^*)$ be a reverse Stackelberg equilibrium and let $f$ be the strategy of $P_i$ in the reverse Stackelberg game,
            then $P_i$ has a strategy in the two-contract game with the same utility for both players: namely,
            $P_i$'s first move is choosing the subgame in which for every second move $s_{-i}$ by $P_{-i}$ they make the cut $f(s_{-i})$. \qed
    \end{description}
\end{proof}

\section{Computational complexity}

Having defined our model of games with smart contracts, in this section we study the computational complexity of computing equilibria in such games. Note that we can always compute the equilibrium by constructing the expanded tree and performing backward induction in linear time.
The problem is that the expanded tree is very large: the expanded tree for a game of size $m$ with a single contract has $2^{O(m)}$ nodes since it contains all possible cuts. For every contract we add, the complexity grows exponentially. This establishes the rather crude upper bound of $\Sigma_k^\EXP$ for computing SPE
in games with perfect information and $k$ contracts. The question we ask if we can do better than traversing the entire expanded tree.

In terms of feasibility, our results are mostly negative: we show a lower bound that computing an SPE, in general, is infeasible for games with smart contracts. We start by considering the case of imperfect information where information sets allow for a rather straightforward reduction from \textsf{CircuitSAT} to games with one contract, showing \NP-completeness for single-contract games of imperfect information. This generalizes naturally to the $k$ true quantified Boolean formula problem (\textsf{$k$-TQBF}), establishing $\Sigma_k^\P$-hardness for games of imperfect information with $k$ contracts. On the positive side, we consider games of perfect information where we provide an algorithm for games and two contracts that runs in time $O(m\ell)$. However, when we allow for an unbounded number of contracts, we show the problem remains \PSPACE-complete by reduction from the generalization of \coloring\, described in \cite{3coloring_pspace}. We conjecture the problem to be \NP-complete for three contracts.

\subsection{Games with imperfect information, \NP-completeness}
We start by showing \NP-completeness for games of imperfect information by reduction from $\textsf{CircuitSAT}$. We consider a decision problem version of SPE: namely, whether or not a designated player can obtain a utility greater than the target value.

\subsubsection{Reduction.} Let $C$ be an instance of \textsf{CircuitSAT}. Note that we can start from any complete basis of Boolean functions, so it suffices to suppose the circuit $C$ consists only of \NAND\,with fanin 2 and fanout 1. We will now construct a game tree for the circuit: we will be using one player to model the assignment of variables, say player 1. The game starts with a contract move for player 1 who can assign values to variables by cutting the bottom of the tree: we construct the game such that player 1 only has moves in the bottom level of the tree. In this way, we ensure that every cut corresponds to assigning truth values to the variables. We adopt the convention that a payoff of 1 for player 1 is \emph{true} ($\top$), while a payoff of 0 for player 1 is \emph{false} ($\bot$). All nodes corresponding to occurrences of the same variable get grouped into the same information set, which enforces the property that all occurrences of the same variable must be assigned the same value.

\begin{figure}
    \centering
    \begin{tikzpicture}[
    level 1/.style={sibling distance=40mm},
    level 2/.style={sibling distance=20mm},
    level 3/.style={sibling distance=10mm}, level distance=1cm,align=center]
    \node[draw] {1}
    child{
        node[] {$\cdots$}
        child {
            node[draw, circle] {1}
            child { node {$\top$} }
            child { node {$\bot$} }}
        child {
            node[draw, circle] (nodeA1) {1}
            child { node {$\top$} }
            child { node {$\bot$} }}}
    child{
        node[] {$\cdots$}
        child {
            node[draw, circle] (nodeA2) {1}
            child { node {$\top$} }
            child { node {$\bot$} }}
        child {
            node[draw, circle] {1}
            child { node {$\top$} }
            child { node {$\bot$} }}};
    \draw[dashed,rounded corners=7]($(nodeA1)+(-.45,-.45)$)rectangle($(nodeA2)+(.45,.45)$);
    \end{tikzpicture}\quad\quad
    \begin{tikzpicture}
    [
    level 1/.style={sibling distance=12mm},
    level distance=1cm,align=center]
    \node {$\vdots$}
    child {
        node[draw,circle] {3}
        child {
            node[draw, circle] {2}
            child { node[draw,isosceles triangle, shape border rotate=90,yshift=-1cm] {$T^L$\\~} }
            child { node {$\bot^\prime$} }
            child { node[draw,isosceles triangle, shape border rotate=90,yshift=-1cm] {\vspace{-0.5cm}$T^R$\\~} } }
        child {
            node {$\top^\prime$}}
    };
    \end{tikzpicture}
    \caption{The basic structure of the reduction. Player 1 has a smart contract that can be used to assign values to the variables. The dashed rectangle denotes an information set and is used when there are multiple occurrences of a variable in the circuit. On the right, we see the \NAND-gate gadget connecting the left subgame $T^L$ and the right subgame $T^R$. We implement the gadget by instantiating the utility vectors such that player 2 chooses $\bot'$ if only if both $T^L$ and $T^R$ propagate a utility vector encoding true.}
\end{figure}
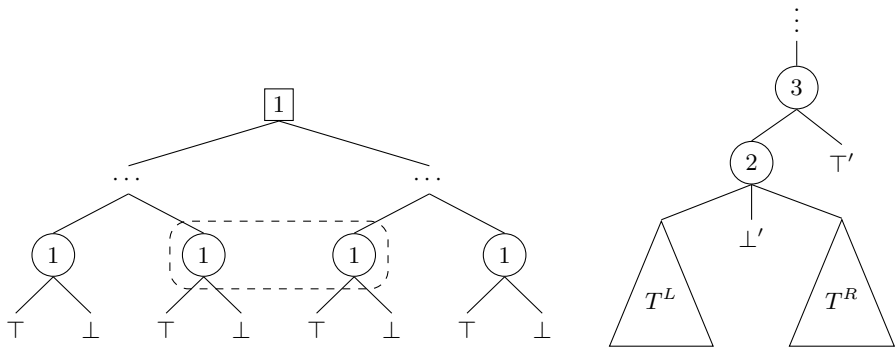

For the \NAND-gate, we proceed using induction: let $T^L, T^R$ be the trees obtained by induction, we now wish to construct a game tree gadget with \NAND-gate logic. To do this we require two players which we call player 2 and player 3. Essentially, player 2 does the logic, and player 3 converts the signal to the right format.  The game tree will contain multiple different utility vectors encoding true and false, which vary their utilities for players 2 and 3. Each \NAND-gate has a left tree and a right tree, each with their own utilities for true and false: $\bot^L, \bot^R; \top^L, \top^R$. The gadget starts with a move for player 3 who can choose to continue the game, or end the game with a true value $\top'$. If they continue the game, player 2 has a choice between false $\bot'$ or playing either $T^L$ or $T^R$. To make the gadget work like a \NAND-gate we need to instantiate the utilities to make backward induction simulate its logic. The idea is to make player 2 prefer both $\bot^L$ and $\bot^R$ to $\bot'$, which they, in turn, prefer to $\top^L$ and $\top^R$. As a result, player 2 propagates $\bot'$ only if both $T^L, T^R$ are true, otherwise, it propagates $\bot^L$ or $\bot^R$. Finally, we must have that player 3 prefers $\top'$ to both $\bot^L$ and $\bot^R$, while they prefer $\bot'$ to $\top', \top^L$ and $\top^R$. This gives rise to a series of inequalities:
\begin{align*}
\bot^L_2 > \bot'_2 > \top^L_2 && \top'_3 > \bot^L_3 && \bot'_3 > \top^L_3  && \bot_3' > \top'_3\\
\bot^R_2 > \bot'_2 > \top^R_2 && \top'_3 > \bot^R_3 && \bot'_3 > \top^R_3
\end{align*}

\noindent We can instantiate this by defining $\top, \bot$. For the base case corresponding to a leaf, we let $\bot = (0,1,0), \top=(1,0,0)$. We then define recursively:
\begin{align*}
\top' &= \left(1,0,1+\max(\top^L_3, \top^R_3)\right)\\
\bot' &= \left(0, \frac{\min(\bot^L_2, \bot^R_2) + \max(\top^L_2, \top^R_2)}2, 2+\max\left(\top_3^L, \top_3^R\right)\right)
\end{align*}

\noindent It is not hard to verify that these definitions make the above inequalities hold true. As a result, the gadget will propagate a utility vector corresponding to true if and only if not both subtrees propagate true.

\begin{theorem}\label{thm:one_contract_np_complete}
    Computing an SPE in three-player single-contract games of imperfect information is $\NP$-complete.
\end{theorem}
\begin{proof}
    We consider the decision problem of determining whether or not in the SPE, player 1 has a utility of 1. By construction of the information sets, any strategy is a consistent assignment of the variables. It now follows that player 1 can get a payoff $>0$ if and only if there is an assignment of the variables such that the output of the circuit is true. This shows \NP-hardness. Now, it is easily seen that this problem is in \NP, since a witness is simply a cut that can be verified in linear time in the size of the tree. Completeness now follows using our reduction from \textsf{CircuitSAT}.\qed
\end{proof}

\begin{remark}
Our reduction also applies to the two-player non-contract case by a reduction from circuit value problem. This can be done in logspace since all the gadgets are local replacements. In doing so, we reestablish the result of \cite{spe_p_complete}, showing that computing an SPE on two-player games is \P-complete. $\hfill\diamond$
\end{remark}

\subsection{Games with imperfect information, \PSPACE-hardness}

In this section, we show that computing the SPE in a game with $k$ contract moves is $\Sigma_k^\P$-complete, in the general case with imperfect information. Generalizing the previous result of \NP-hardness to $k$ contracts is fairly straightforward. Our claim is that the resulting decision problem is $\Sigma_k^\P$-hard so we obtain a series of hardness results for the polynomial hierarchy. This is similar to the results obtained in \cite{polyhierarchy} where the value problem for a competitive analysis with $k+1$ players is shown to be hard for $\Sigma_k^\P$.

Formally, we consider the following decision problem with target value $V$ for a game tree $T$ with $k$ contract players: let $T^\prime$ be the expanded tree with contracts for players $P_1, P_2, \ldots P_k$ in ascending order. Can player $P_1$ make a cut in $T^\prime$ such that their payoff is $\geq V$?

To show our claim, we proceed using reduction from the canonical $\Sigma_k^\P$-complete problem \textsf{$k$-TQBF}, see e.g. \cite{computers_and_intractability} for a formal definition.

\begin{theorem}\label{thm:poly_hierarchy-hard}
    Computing an SPE in $2+k$ player games of imperfect information is $\Sigma_k^\P$-hard.
\end{theorem}
\begin{proof}[sketch]
    We extend our reduction from \cref{thm:one_contract_np_complete} naturally to the quantified satisfiability problem. In our previous reduction, the contract player wanted to satisfy the circuit by cutting as to assign values to the variables in the formula. Now, for each quantifier in $\psi$, we add a new player with a contract, whose moves range over exactly the variables quantified over. The players have contracts in the same order specified by their quantifiers. The idea is that players corresponding to $\forall$ try to sabotage the satisfiability of the circuit, while those corresponding to $\exists$ try to ensure satisfiability. We encode this in the utility vectors by giving $\exists$-players a utility of $1$ in $\top$ and 0 utility in $\bot$, while for the $\forall$-players, it is the opposite. It is not hard to see that $\psi$ is true, only if $P_1$ can make a cut, such that for every cut $P_2$ makes, there exists a cut for $P_3$ such that, ..., the utility of $P_1$ is 1. This establishes our reduction. \qed
\end{proof}

\begin{remark}
    We remark that it is not obvious whether or not the corresponding decision problem is contained within $\Sigma_k^\P$. It is not hard to see we can write a Boolean formula equivalent to the smart contract game in a similar manner as with a single contract. The problem is that it is unclear if the innermost predicate $\phi$ can be computed in polynomial-time. It is not hard to see that some smart contracts do not have a polynomial description, i.e. we can encode a string $x \in \{0,1\}^*$ of exponential length in the contract. However, there might be an equivalent contract that \emph{does} have a polynomial-time description. By equivalent, we mean one that has the same dominating path. This means that whether or not $\Sigma_k^\P$ is also an upper bound essentially boils down to whether or not every contract has an equivalent contract with a polynomial description. $\hfill\diamond$
\end{remark}

\subsection{Games with perfect information, two contracts, upper bound}
In this section, we consider two-player games of perfect information and provide a polynomial-time algorithm for computing an SPE in these games.
Specifically, for a game tree of size $m$ with $\ell$ terminal nodes with two contract players (and an arbitrary number of non-contract players), we can compute the equilibrium in time $O(m\ell)$. Our approach is similar to that of \cite{inducibleregion}, in that we compute the inducible region for the first player, defined as the set of leaves they are able to `induce' by making cuts in the game tree.

Let $A,B$ be two sets. We then define the set of outcomes from $A$ reachable using a threat against player $i$ from outcomes in $B$ as follows:
\begin{equation*}
    \label{eq:threaten}
    \text{threaten}_i(A,B) = \{x \in A \mid \exists\, y \in B.\, x_i > y_i\}
\end{equation*}
As mentioned, we will compute the \emph{inducible region} for the player with the first contract, defined as the set of outcomes reachable with a contract. Choosing the optimal contract is then reduced to a supremum over this region.
\begin{definition}
    Let $G$ be a fixed game. We denote by $\mathscr{R}(P_1)$ (resp. $\mathscr{R}(P_1,P_2)$) the \emph{inducible region of $P_1$}, defined as the set of outcomes reachable by making a cut in $G$ in all nodes owned by $P_1$. $\mathscr{R}(P_1)$ is a tuple $(\mathbf{u}, c_1)$ where $\mathbf{u} \in \mathbb{R}^n$ is the utility vector, and $c_1$ is the contract (a cut) of $P_i$. $\hfill\diamond$
\end{definition}

\subsubsection{Algorithm.} Let $G$ be the game tree in question and let $k$ be a fixed integer. As mentioned, we assume without loss of generality that $G$ is in \emph{generic form}, meaning all non-leafs in $G$ have out-degree exactly two and that all utilities for a given player are distinct such that the ordering of utilities is unique. We denote by $P_1, P_2$ the players with contracts and assume that $P_i$ has the $i^\text{th}$ contract. We will compute the inducible regions in $G$ for $P_1$ (denoted $S$ for \emph{self}), and for $(P_1,P_2)$ (denoted $T$ for \emph{together}) by a single recursive pass of the tree. In the base case with a single leaf with label $\mathbf{u}$ we have $S=T=\{\mathbf{u}\}$. For a non-leaf, we can recurse into left and right child, and join together the results. The procedure is detailed in \cref{algo:inducibleregion}.

\begin{algorithm}
    \SetAlgoLined
    \Switch{$G$}{
        \textbf{case} $\texttt{Leaf}(u)$:\\\Indp\Return{$(\{u\},\{u\})$}\\~\\\Indm

        \textbf{case} $\texttt{Node}(G^L, G^R, i)$:\\\Indp
        $(S^L, T^L) \gets \texttt{InducibleRegion}(G^L)$\\
        $(S^R, T^R) \gets \texttt{InducibleRegion}(G^R)$\\
        \uIf{$i=1$}{
            $T \gets T^L \cup T^R$\\
            $S \gets S^L \cup S^R \cup \text{threaten}_{2}(T^L \cup T^R, S^L \cup S^R) $
        }
        \uElseIf{$i=2$}{
            $T \gets T^L \cup T^R$\\
            $S \gets \text{threaten}_{2}(T^L, S^R) \cup \text{threaten}_{2}(T^R, S^L) $
        }
        \Else{
            $T \gets \text{threaten}_{i}(T^L, T^R) \cup \text{threaten}_{i}(T^R, T^L)$\\
            $S' \gets \text{threaten}_{i}(S^L, S^R) \cup \text{threaten}_{i}(S^R, S^L)$\\
            $S \gets S' \cup \text{threaten}_2(T,S')$
        }
        \Return{$(S,T)$}
    }

    \caption{$\texttt{InducibleRegion}(G)$}
    \label{algo:inducibleregion}
\end{algorithm}

\begin{theorem}\label{thm:two_contracts_feasibility}
    An SPE in two-contract games of perfect information can be computed in time $O(m\ell)$.
\end{theorem}
\begin{proof}
    First, the runtime is clearly $O(m\ell)$ since the recursion has $O(m)$ steps where we need to maintain two sets of size at most $\ell$. For correctness, we show something stronger: let $\inducible(P_1)$ be the inducible region for $P_1$ in the expanded tree and $\inducible(P_1, P_2)$ be the inducible region of $(P_1, P_2)$. Now, let $(S, T) = \texttt{InducibleRegion}(G)$.
    Then we show that $S = \inducible(P_1)$ and $T = \inducible(P_1,P_2)$. This implies that $\argmax_{u \in S} u_1$ is the SPE. The proof is by induction on the height $h$ of the tree. As mentioned, we assume that games are in \emph{generic form}. This base case is trivial so we consider only the inductive step.

    Necessity follows using simple constructive arguments: for $S$ and $i=1$, then for every $(\mathbf{u},c) \in S^\ell$, we can form contract where $P_{1}$ chooses left branch and plays $c$. And symmetrically for $S^R$. Similarly, for every $(\mathbf{u},c_1,c_2) \in T^L$ and $(\mathbf{v},c') \in S^L$ can form contract where $P_{1}$ plays $c_1$ in all subgames where $P_{2}$ plays $c_2$; and plays $c'$ otherwise. Then $\mathbf{u}$ is dominating if and only if $\mathbf{u}_2 > \mathbf{v}_2$. Similar arguments hold for the remaining cases.

    \noindent For sufficiency, we only show the case of $i=1$ as the other cases are similar. Assume (for contradiction) that there exists $(\mathbf{u}, c_1) \in \inducible(P_1) \setminus S$, i.e. there is a $P_1$-cut $c_1$ such that $\mathbf{u}$ is dominating. Then,
        \begin{align*}
        (\mathbf{u},c_1) &\in (T^L \cup T^R) \setminus (S^L \cup S^R \cup \text{threaten}_2(T^L \cup T^R, S^L \cup S^R)) \\
        &= \{ \mathbf{v} \in (T^L \cup T^R) \setminus (S^L \cup S^R) \mid \forall \mathbf{v}' \in S^L \cup S^R.\, \mathbf{v}_2 < \mathbf{v}'_2\}
        \end{align*}
    That is, $\mathbf{u}$ must be a utility vector that $P_1$ and $P_2$ can only reach in cooperation in a one of the two sub-games, say by $P_2$ playing $c_2$. However, for every cut that $P_1$ makes, the dominating path has utility for $P_2$ that is $> \mathbf{u}_2$, meaning $P_2$ strictly benefits by not playing $c_2$. But this is a contradiction since we assumed $\mathbf{u}$ was dominating.\qed
\end{proof}

\subsection{Games with perfect information, unbounded contracts, \PSPACE-hardness}

We now show that computing an SPE remains \PSPACE-complete when considering games with an arbitrary number of contract players.
We start by showing \NP-hardness and generalize to \PSPACE-hardness in a similar manner as we did for \cref{thm:poly_hierarchy-hard}. The reduction is from \coloring: let $(V,E)$ be an instance of \coloring\, and assume the colors are $\{R,G,B\}$.
The intuition behind the \NP-reduction is to designate a coloring player $\Pcolor$,
who picks colors for each vertex $u \in V$ by restricting his decision space in a corresponding move using a contract.
They are the first player with a contract.
This is constructed using a small stump for every edge $e \in E$ with three leaves $R_u, G_u, B_u$.
We also have another player $\Pcheck$ whose purpose is to ensure no two adjacent nodes are colored the same.
We attach all stumps to a node owned by $\Pcheck$ such that $\Pcheck$ can choose among the colors chosen by $\Pcolor$.
If $\Pcolor$ are able to assign colors such that no adjacent nodes share a color,
then $\Pcolor$ maximizes their utility,
however, if no such coloring exists then $\Pcheck$ can force a bad outcome for $\Pcolor$.
It follows that $\Pcolor$ can obtain good utility if and only there is a valid coloring.

\begin{figure}\label{fig:reduction}
	\centering
	\scalebox{0.85}{
		\begin{tikzpicture}
		[
		level 8/.style={sibling distance=25mm},
		level 9/.style={sibling distance=8mm},
		level distance=1.4cm,align=center,
		emph/.style={edge from parent/.style={ very thick,draw}},
		norm/.style={edge from parent/.style={solid,black,thin,draw}}]]
		
		\node[draw] {$\Pcheck$}
		child {
			node[draw,yshift=0.5cm] {$\Pcolor$}
			child {
				node[draw,yshift=0.5cm] {$P_{u_1,R}$}
				child {
					node[draw,yshift=0.5cm] {$P_{u_2,R}$}
					child {
						node[yshift=0.5cm] {$\vdots$}
						child {
							node[draw,circle,yshift=0.25cm] {$P_{u_1,R}$}
							child[emph] {
								node {$\vdots$}
								child[norm] {
									node[draw,thin,circle,yshift=0.2cm] {$\Pcheck$}
									child {
										node[draw,circle] {$\Pcolor$}
										child { node {$R_{u_1}$}}
										child { node {$G_{u_1}$}}
										child { node {$B_{u_1}$}}
									}
									child {
										node[draw,circle] {$\Pcolor$}
										child { node {$R_{u_2}$}}
										child { node {$G_{u_2}$}}
										child { node {$B_{u_2}$}}
									}
									child{
										node {$\cdots$}
									}
									child {
										node[draw,circle] {$\Pcolor$}
										child { node {$R_{u_n}$}}
										child { node {$G_{u_n}$}}
										child { node {$B_{u_n}$}}
									}
								}
							}
							child {
								node[draw,circle,xshift=2cm] {$P_{u_2,R}$}
								child[emph] {
									node[yshift=0.25cm,xshift=-0.2cm] {$\top_{(u_1,u_2),R}$}
								}
								child {
									node[yshift=0.25cm,xshift=0.2cm] {$\bot_{(u_1,u_2),R}$}
								}
							}
		}}}}};
		\end{tikzpicture}}
	\caption{The structure of the reduction. First, $\Pcolor$ is allowed to assign a coloring of all vertices. If there is no \coloring\, of the graph, there must be some vertex $(u_1,u_2)$ where both vertices are colored the same color $c$. In this case, $\Pcheck$ can force both $c_{u_1}, c_{u_2}$, which are undesirable to $P_{u_1,c},$ resp. $P_{u_2,c}$: then in every $P_{u_1,c}$-contract where they do not commit to choosing $P_{u_2,c}$, $\Pcheck$ cuts as to ensure $c_{u_1}$ and analogously for $P_2$. It follows that $\Pcheck$ can get $\bot$ if and only if the graph is not 3-colored. Then $\Pcolor$ can get a different outcome from $\bot$ if and only if they can 3-color the graph.}
	\label{fig:my_label}
\end{figure}
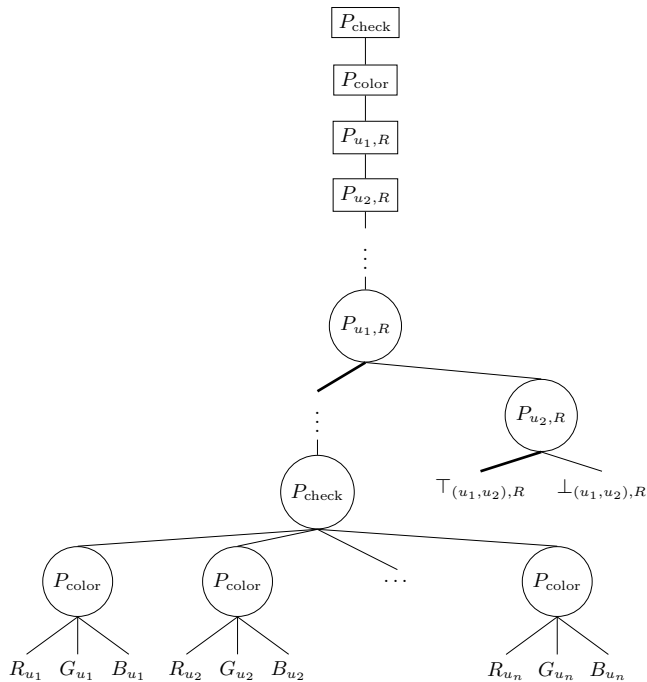

\subsubsection{Reduction.} We add six contract players for every edge in the graph. Specifically, for every edge $(u,v) \in E$ and every color $c \in \{R,G,B\}$, we introduce two new contract players $P_{u,c}$ and $P_{v,c}$ who prefer any outcome except $c_u$ (resp. $c_v$) being colored $c$. That is, if $c=R$, then the leaf $R_u$ has a poor utility for $P_{u,R}$. We add moves for $P_{u,c}$ and $P_{v,c}$ at the top of the tree, such that if they cooperate, they can get a special utility vector $\bot_{u,v}$ which has a poor utility for $\Pcolor$ and great utility for $\Pcheck$, though they themselves prefer any outcome in the tree (except $c_u$, resp. $c_v$) to $\bot_{u,v}$. We ensure that $\Pcheck$ has a contract directly below $\Pcolor$ in the tree. If no coloring exists, then $\Pcheck$ can force a bad outcome for both $P_{u,c}, P_{v,c}$ in all contracts where they do not commit to choosing $\bot_{u,v}$. Specifically, $\Pcheck$ first threatens $P_{u,c}$ with the outcome $c_u$, and subsequently threatens $P_{v,c}$ with $c_v$. Though they prefer any other node in the tree to $\bot_{u,v}$, they still prefer $\bot_{u,v}$ to $c_u$, $c_v$, meaning they will comply with the threat. This means $\Pcolor$ will receive a poor outcome if the coloring is inconsistent. It follows that $\Pcolor$ will only receive a good payoff if they are able to 3-color the graph, see e.g \cref{fig:reduction} for an illustration.

\begin{theorem}\label{thm:perfect_information_pspace_hard}
    Computing an SPE in smart contract games of perfect information is $\PSPACE$-hard when we allow for an unbounded number of contract players.
\end{theorem}
\begin{proof}
    Let $(V,E)$ be an instance of \coloring. Our above reduction works immediately for $k=1$, showing \NP-hardness. To show \PSPACE-hardness we reduce from a variant of \coloring\,as described in \cite{3coloring_pspace} where players alternately color an edge and use a similar trick as \cref{thm:poly_hierarchy-hard} by introducing new players between $\Pcolor$ and $\Pcheck$.\qed
\end{proof}

\noindent It remains unclear where the exact cutoff point is, though we conjecture it to be for three contracts: clearly, the decision problem for three-contract games of perfect information is contained in \NP\, as the witness (a cut for the first contract player) can be verified by Algorithm \ref{algo:inducibleregion}.

\begin{conjecture}\label{conj:npc}
    Computing an SPE for three-contract games is \NP-complete. $\hfill\diamond$
\end{conjecture}

\section{Conclusion}
In this paper, we proposed a game-theoretic model for games in which players have shared access to a blockchain that allows them to deploy smart contracts. We showed that our model generalizes known notions of equilibria, with a single contract being equivalent to a Stackelberg equilibrium and two contracts equivalent to a reverse Stackelberg equilibrium.
We proved a number of bounds on the complexity of computing an SPE in these games with smart contracts, showing, in general, it is infeasible to compute the optimal contract.

\bibliographystyle{splncs04}
\bibliography{refs}

\end{document}